\documentclass{article}
\usepackage[utf8]{inputenc}
\usepackage{amsfonts}
\usepackage{amssymb}
\usepackage{amsthm}
\usepackage{amsmath}
\RequirePackage{graphicx}

\title{Infinitude of Primes Using Formal Languages}
\author{Aalok Thakkar}
\date{}

\newtheorem{theorem}{Theorem}
\newtheorem{proposition}{Proposition}

\newtheorem{definition}{Definition}
\newtheorem{remark}{Remark}
\newtheorem{example}{Example}

\begin{document}

\maketitle

\begin{abstract}
Formal languages are sets of strings of symbols described by a set of rules specific to them. In this note, we discuss a certain class of formal languages, called regular languages, and put forward some elementary results. The properties of these languages are then employed to prove that there are infinitely many prime numbers. 
\end{abstract}

\section{Introduction.} If you have used the command line shell of an operating system, or search facilities of text editors, undoubtedly you would have come across pattern matching by entering an expression to match with text. These expressions can be used to describe sets of words. The computer can then simulate a simple model of computation to decide the membership of a given string in these sets. This note presents a formal characterization of such expressions and the sets they represent. Properties of these sets are discussed and then used to prove the infinitude of primes.

\section{Words and Languages.} In order to formalize the notion of expressions, we first look at alphabets, words, and languages. An alphabet $\Sigma$ is a set of formal symbols. A priori, these symbols do not bear any relation to each other (such as order). A finite word of length $k$ is a concatenated sequence $a_1 \cdot a_2 \cdot \ldots \cdot a_k $ where $a_i \in \Sigma$. An empty word is denoted by the symbol $\epsilon$. 

\begin{remark}
In the following sections, we only deal with finite alphabet and finite-length words.
\end{remark}

The set of all words over an alphabet $\Sigma$ is denoted by $\Sigma^*$. One can think of $\Sigma^*$ as the free monoid generated by elements of $\Sigma$, with $\epsilon$ as the identity, and concatenation as the binary operation. Using formal rules, one can describe subsets of $\Sigma^*$ which are called languages. The nomenclature reflects the idea that letters of the alphabet make words and words make languages. Formal language theory deals with the study of these languages, their properties, their representations, and computations using them. The following are three operations central to formal languages.

\begin{itemize}
\item (Concatenation) $L_1 \cdot L_2 = \{ u \cdot v : u \in L_1 \text{ and } v \in L_2 \}$ denotes the set of strings obtained by concatenating a string in $L_1$ with a string in $L_2$. 

\item (Union) $L_1 + L_2 = L_1 \cup L_2 = \{ w : w \in L_1 \text{ or } w \in L_2 \}$ denotes the union of $L_1$ and $L_2$.

\item (Kleene star) $L^* = \cup_{k \geq 0} L^k $ where $L^k$ is the $k$ times repeated concatenation of $L$. By definition, $L^0 = \{\epsilon\}$. One can think of $L^*$ as the smallest set containing $\epsilon$ and $L$ that is closed under concatenation, or as the free monoid generated by the elements of $L$.
\end{itemize}

\section{Regular Expressions and Regular Languages.} In general, a language may be infinite, in which case it is necessary to look for a finite representation of it. Expressions offer one such way to represent a certain class of languages. Given an alphabet $\Sigma$, an expression is a finite-length string of characters that uses symbols from $\Sigma$ and operators to describe the language. For notation, if $e$ is an expression, then $\mathcal{L} (e)$ denotes the language represented by $e$. Of the many equivalent ways to formally describe regular expressions, we opt for the one preferred by most introductory references for formal language theory. \\

The following are the constant expressions.
\begin{itemize}
\item $\emptyset$ denotes the empty set.
\item $a_i$ denotes the singleton set containing $a_i$ (a character in the alphabet $\Sigma$).
\end{itemize}

One can define the three operations of concatenation, union, and Kleene star for the expressions analogously. Note that $\emptyset^* = \{\epsilon\}$ by the definition of Kleene star.

\begin{definition}[Regular Expression]
The set of \emph{regular expressions} is the smallest set closed under concatenation, union, and Kleene star that contains the constant expressions. An expression is said to be \emph{regular} if it belongs to the set of regular expressions. 
\end{definition}

The expressions used for pattern matching in command line shells are exactly these regular expressions (sometimes with additional operators such as those to match with the beginning or the end of a line). The priority of operation is first given to Kleene star, followed by union, then by concatenation. Therefore, $a + b \cdot c^*$ means $a + (b \cdot (c^*))$. The following examples demonstrate the use of expressions to represent languages.

\begin{example}
Let $\Sigma = \{a, b, c\}$. The language $L \subset \Sigma^*$ consisting of words with at least one $a$ can be represented by the expression $ (a+b+c)^* \cdot a \cdot (a+b+c)^*$, and the language $L' \subset \Sigma^*$ consisting of words that do not contain any $a$ can be represented by $(b + c)^*$.
\end{example}

\begin{example}
Let $\Sigma = \{a, b, c\}$. Let $L''$ be the set of strings $w \in \Sigma^*$ such that for every $a$ occurring in $w$, there is at least one $b$ occurring to the right of it. For a word $w$ in $L''$, if $a$ does not occur in $w$, then it belongs to the language represented by $(b + c)^*$. Otherwise, $a$ occurs in $w$, and so $w$ can be written as $u \cdot b \cdot v$ where $u \in \Sigma^*$ contains at least one $a$, and $v \in \Sigma^*$ does not contain an $a$. Hence, $L''$ can be represented by the expression $(b + c)^* + (a+b+c)^* \cdot a \cdot (a+b+c)^* \cdot b \cdot (b + c)^*$.
\end{example}

Expressions are used to represent languages as they are concise and simple to give as input to computers. In order to check if a given string $w$ matches the expression $e$, that is, $w \in \mathcal{L}(e)$, the computer needs to allocate some bounded memory based on the expression (independent of the length of the word), and then the membership can be checked in time that is linear in the length of the word. The model of computation simulated for this process is called a \textit{deterministic finite automaton} (DFA). Informally, a DFA is a nonempty set of states $Q$, of which exactly one is an initial state and zero or more are final states, together with a transition function $\delta: Q \times \Sigma \rightarrow Q$. DFAs also represent languages, and a word is said to belong to the language of a given DFA if successive application of the transition function on reading the letters of the word map the initial state to one of the final states. Continuing Example 1, the DFA representing $L$ needs only two states as shown in Figure 1.

\begin{figure}
    \centering
        \includegraphics[]{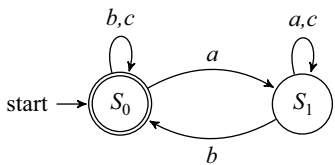}
        \caption{A DFA accepting the language described in Example 1. $S_0$ is the initial state (as indicated by the transition from "start") and the only final state (as indicated by the double-circle). Starting with $S_0$, we remain in $S_0$ on reading $b$ or $c$, and move to $S_1$ on reading an $a$. We remain in $S_1$ on reading $a$ or $c$, and move to $S_0$ on reading a $b$. The word is accepted if and only if we are in $S_0$ when the word ends.}
        \label{fig1}
\end{figure}

In 1956, Kleene proved the equivalence of the set of languages represented by regular expressions and the set of languages recognized by a DFA. We use this equivalence to define regular languages.

\begin{definition}[Regular Languages] $L \subset \Sigma^*$ is \emph{regular} if and only if it can be represented by a regular expression. 
\end{definition}

In order to show that the definition is meaningful, we give an example of a language that is not regular. As $\Sigma^*$ is a countably infinite set, its power set is uncountable. If one shows that there are only countably many regular expressions, the result is immediate. The reader is left to fill in the details while we take a constructive approach.

\begin{example}
Let $\Sigma = \{a\}$. Let $L$ be the set of strings $w \in \Sigma^*$ such that the length of $w$ is a power of 2. Let us show that $L$ is not regular. For the sake of contradiction, suppose there is a regular expression $e$ that represents $L$. As the set $L$ is infinite, the Kleene star operation must be used somewhere in the expression. The expression can be decomposed as $e = e' + (e_P \cdot (e_Q)^* \cdot e_R)$, where $e_P$, $e_Q$, and $e_R$ are non-empty and $e_Q$ is not equal to $\{\epsilon\}$. Let $w = (a^p) \cdot (a^q) \cdot (a^r)$ be such that $a^p \in \mathcal{L}(e_P)$, $a^q \in \mathcal{L}(e_Q^*)$, and $a^r \in \mathcal{L}(e_R)$. As $\mathcal{L}(e_Q^*)$ is infinite, we can chose $q > 0$. By definition of Kleene star, for all $n \in \mathbb{N}$, $a^{nq} \in \mathcal{L}(e_q^*)$. Hence, $ a^{p + nq + r} \in \mathcal{L}(e_P \cdot (e_Q)^* \cdot e_R) \subset \mathcal{L}(e)$. Therefore, for $q \neq 0$ and for some $p$ and $r$, $(p + r + nq)$ is a power of 2 for all $n \in \mathbb{N}$. We have produced an arithmetic progression in the set of powers of 2. Elementary number theory tells us that there is no arithmetic progression in the set of powers of $2$. Hence there is a contradiction. 
\end{example}

In order to prove irregularity using only the expressions, one needs to work combinatorially. This becomes involved when the alphabet is not singleton, and the language does not have such a simple structure. In such cases, distinguishing extensions provide an alternative. 

\begin{definition}[Distinguishing Extension] Given a language $L \subset \Sigma^*$, and a pair of strings $x$ and $y$ in $\Sigma^*$, a \emph{distinguishing extension} is a string $z \in \Sigma^*$ such that exactly one of the two strings $xz$ or $yz$ is a member of $L$. 
\end{definition}

The distinguishing extensions partition $\Sigma^*$ into equivalence classes that describe certain properties of the language. For regular languages, we have the following.

\begin{theorem}[Myhill--Nerode]
\label{nerode} Let $x \equiv_L y$ if there is no distinguishing extension for $x$ and $y$ with respect to $L$. Then
$L$ is regular if and only if $\equiv_L$ induces finitely many equivalence classes.
\end{theorem}

Theorem 1 is a strong and useful characterization of regular languages and provides a systematic way to deal with regularity testing. Other alternatives are the pumping lemma, Parikh's theorem, and tests based on the closure properties . The following examples show the application of Theorem \ref{nerode}.

\begin{example}
Let $\Sigma = \{a,b\}$ and let $|w|$ denote the length of a word $w$. Let $L$ be the set of all strings $w \in \Sigma^*$ with length of the form $5k + 3$, that is, $|w| \equiv 3 \mod 5$. Given $u, v \in \Sigma^*$ such that $m \equiv |u| \not\equiv |v| \mod 5$ and $0 \leq m < 4$, let $w = a^{8 - m}$. Notice that $|u\cdot w| = m + 8 - m \equiv 3 \mod 5$ while $|v \cdot w| \not\equiv m + 8 - m \equiv 3 \mod 5$. Hence $w$ distinguishes $u$ and $v$. Also, if $|u| \equiv |v| \mod 5$, then there is no distinguishing string and hence $u \equiv_L v$. This implies that, $\equiv_L$ induces exactly five equivalence classes, partitioning $\Sigma^*$ into the classes based on the length of a string modulo five, which is to say that, given a string $w$, it must be equivalent to $a$, $a^2$, $a^3$, $a^4$, or $a^5$ under $\equiv_L$. Therefore, by Theorem 1, $L$ is regular.
\end{example}

\begin{example}
Let $\Sigma = \{a\}$ and let $L$ be the set of all strings $w \in \Sigma^*$ with length equal to a Fibonacci number. Let $F(i)$ be the $i$th Fibonacci number. Let $A = \{ a^{F(3k)} : k \in \mathbb{N}; F(k) > 1\}$. For two distinct elements $x = a^{F(3i)}$ and $y = a^{F(3j)}$ of $A$ with $i > j$, $z = a^{F(3i-1)}$ is a distinguishing extension as $xz \in L$ and $yz \not\in L$, by the definition of Fibonacci numbers. Hence no two elements of $A$ can belong to the same equivalence class, and by Theorem 1, we have that $L$ is not regular.
\end{example}

Not only does Theorem 1 put forward a rigorous test of regularity, but it also brings out the algebraic character of formal language theory. The concept of distinguishing extensions was further developed by M. P. Sch\"{u}tzenberger in his seminal paper on star-free languages, and by Krohn and Rhodes in their work on an algebraic theory of machines. One can appreciate the parallels between algebra and computation---the spirit of formalism, the affection for abstraction, and the strive for elegance. Encompassing logic and proof theory, combinatorics and computation complexity, the domain of the intersection of the two subjects also provides the following proof of infinitude of primes.

\section{A Class of Regular Languages.} In order to prove the infinitude of primes, we look at a particular class of languages. Consider the alphabet $\Sigma = \{a,b \}$. Let $|w|_{\alpha}$ denote the number of occurrences of the character $\alpha$ in the string $w$, and set $\xi(w) = |w|_a - |w|_b$. For $n \in \mathbb{Z}^+$, let

$$L_n = \{ w \in \Sigma^* : \xi(w) \text{ is divisible by } n \}.$$

\begin{proposition}
For all $ n \in \mathbb{Z}^+$, the language $L_n$ is regular.
\end{proposition}

\begin{proof} Similar to Example 3, given $u, v \in \Sigma^*$, let $m \equiv \xi(u) \not \equiv \xi (v) \mod n$ where $0 \leq m< n$. The string $a^{n - m}$ distinguishes $u$ from $v$. Indeed, note that $\xi$ is additive, which is to say that $\xi(u \cdot v) = \xi(u) + \xi(v)$. Now, $\xi(ua^{n-m}) = \xi(u) + \xi(a^{n-m}) \equiv m + n - m \equiv 0 \mod n$, which implies that $ua^{n-m} \in L_n$. A similar calculation shows that $va^{n-m} \not\in L_n$. If $\xi(u) \equiv \xi (v) \mod n$, then there is no distinguishing extension. Therefore, $\equiv_{L_n}$ induces exactly $n$ equivalence classes, that is, any word $w \in \Sigma^*$ is equivalent to one of the elements of the set $\{a^i : 0 \leq i < n\}$. By Theorem 1, $L_n$ is regular.
\end{proof}

Let $\mathbb{P} \subset \mathbb{N}$ be the set of primes. Let

\begin{equation}
\label{L}
L = \bigcup\limits_{p \in \mathbb{P}}L_p.
\end{equation}

\begin{proposition}
$L = \{ w \in \Sigma^* : \xi(w) \neq \pm 1 \}$.
\end{proposition}

\begin{proof}
By definition, for every integer other than $+1$ and $-1$, there exists a prime number dividing it, and there is no prime number dividing $+1$ or $-1$. For $w \in \Sigma^*$, if $\xi(w) \neq \pm 1$, then some $p \in \mathbb{P}$ divides $\xi(w)$ and if $\xi(w) = \pm 1$, then no prime in $\mathbb{P}$ divides $\xi(w)$. Therefore $w \in \cup_{p \in \mathbb{P}} L_p$ if and only if $\xi(w) \neq \pm 1$.
\end{proof}

\begin{proposition}
\label{proposition3}
$L$ is not regular.
\end{proposition}

\begin{proof}
Consider the set $A = \{a^{3k} : k \in \mathbb{N} \}$. For $x = a^{3i}$ and $y = a^{3j}$ in $A$ with $i > j$, the string $z = b^{3j+1}$ distinguishes the two as $\xi(xz) = \xi(a^{3i}b^{3j+1}) = 3(i - j) - 1 \geq 2$ which implies  $xz \in L$, whereas $\xi(yz) = \xi(a^{3j}b^{3j + 1}) = -1$ so $yz \not\in L$. Therefore, no two elements of the infinite set $A$ belong to the same equivalence class. By Theorem 1, $L$ is not regular.
\end{proof}

\section{Infinitude of Primes.} By the definition of regular expressions, the union of two regular languages is regular, and hence regularity is preserved under finite union. For the sake of contradiction, suppose there are only finitely many primes; then we have $L$ as a finite union of regular languages $L_p$ making $L$ regular. This contradicts Proposition \ref{proposition3}. Hence we have shown that there must be infinitely many prime numbers. A motivated reader may now assume the infinitude of primes and prove that the language of words of prime length is not regular.

\vfill\eject

\end{document}